\documentclass[12pt]{amsart}
\usepackage{amsfonts}

\usepackage[letterpaper]{geometry}
\usepackage{amsmath}
\usepackage{mathtools}
\usepackage{commath}

\usepackage{amssymb}
\def\N{\mathbb{N}}

\def\bo{\{0,1\}^n}

\newtheorem{thm}{\bf Theorem}[section]
\newtheorem{lemma}[thm]{\bf Lemma}
\newtheorem{prop}[thm]{\bf Proposition}

\theoremstyle{definition}

\newtheorem*{rem}{\bf Remark}

\begin{document}

\title[On a Communication Complexity problem in Combinatorial Number Theory]{On a Communication Complexity problem in Combinatorial Number Theory}

\author[]{Bence Bakos}
\address{Bence Bakos, ELTE TTK,
E\"otv\"os University, Institute of Mathematics, H-1117
P\'{a}zm\'{a}ny st. 1/c, Budapest, Hungary}
 \email{bakosbence237@gmail.com}

\author[]{Norbert Hegyv\'ari}
\address{Norbert Hegyv\'{a}ri, ELTE TTK,
E\"otv\"os University, Institute of Mathematics, H-1117
P\'{a}zm\'{a}ny st. 1/c, Budapest, Hungary and Alfr\'ed R\'enyi Institute of Mathematics, H-1364 Budapest, P.O.Box 127.}
 \email{hegyvari@renyi.hu}

\author[]{M\'at\'e P\'alfy}
\address{M\'at\'e P\'alfy, ELTE TTK,
E\"otv\"os University, Institute of Mathematics, H-1117
P\'{a}zm\'{a}ny st. 1/c, Budapest, Hungary}
 \email{palfymateandras@gmail.com}

\large

\maketitle


\begin{abstract}

The original knapsack problem is well known to be NP-complete. In a multidimensional version one have to decide whether a $p\in \N^k$ is in a sumset-sum of a set $X \subseteq \N^k$ or not. In this paper we are going to investigate a communication complexity problem related to this. Throughout our examination we are also going to prove some results about the special case of the multidimensional knapsack problem, when the set $X$ is in the form $X=A_1 \times \dots \times A_k \subseteq \N^k$, where $A_i$ are so called \textit{regular} sets.

AMS 2010 Primary 11B30, 11B39, Secondary 11B75

\end{abstract}

\maketitle

\section{Introduction}

In the last decades there were several interplay between computer sciences  and additive combinatorics. One of the most interesting examples is a connection between some notions in computer sciences and the {\it Gowers norm} (see e.g. [8], [9]).
Another interesting example is an additive communication complexity problem which was supported by an example of Behrend on the maximal density of a set without a three-term arithmetic progression (see e.g. [7]) and additive problem related to decision question (see [6]). 

In our paper we investigate a communication complexity problem which is related to a field in combinatorial number theory; namely to the topic of  subset-sums (see e.g. [5] and for restricted set addition e.g. [4]).

The original form of the knapsack or subset-sum problem is to determine, given positive integers $a_1,a_2,\dots,a_n$ and an integer $m$, whether there is a subset of the set $\{a_j\}$ that sums up to $m$. It is well known as an NP-complete problem. Certainly this question can be extend to higher dimension too. For any $X\subseteq \N^k$ let
\begin{equation}\label{1}
FS(X):=\{\sum_{i=1}^\infty\varepsilon_ix_i: \ x_i\in X, \ \varepsilon_i \in \{0,1\}, \ \sum_{i=1}^\infty\varepsilon_i<\infty\}\end{equation}
and we have to decide for $p\in \N$ whether $p\in FS(X)$ or not.

The structure of $FS(X)$ in higher dimension was investigated in $[1],[2]$ and $[3]$ when the sets $A_i$ are powers of integers.

Surprisingly if we have $X=A_1\times A_2\times \dots \times A_k\subseteq \N^k$, where for any $i=1,2,\dots,k$,  $FS(A_i)=\N$ then $FS(X)$ does not cover necessary the whole $\N^k$. For example if $X=\{2^k\}_{k=0}^\infty\times \{2^m\}_{m=0}^\infty$ then $(15,1)\not \in FS(X)$, while $(15,1+256) \in FS(X)$.

Our communication complexity problem can be described as follows: there are $k$ players, each having a sequence $A_i$ with following properties: for every $i=1,2,\dots,k$
$$
(i) \ 1\in A_i; \quad (ii) \ A_i\setminus \{1\}\subseteq A_i+A_i; \quad (iii) \ a_{j+1}>\varrho a_j
$$
for some  $\varrho\in (1,2]\ j\in \N; \ a_j\in A_i$. 

For sets $A,B\subset \N$ the sum is defined by $A+ B:=\{a\pm b:a\in A; \ b\in B\}$, and throughout the paper $\log_2 N$ will denote the logarithm in base 2.

The set $Y\subseteq \N$ is said to be \textit{complete} if $FS(Y)=\N$. For any $z\in \N^r$ and $X\subseteq \N^r$ let us denote by $r(z)$ the representations of $z$ from $X$, i.e.
$r(z):=r_X(z)=\{(\varepsilon = \{\varepsilon_i\}_{i=1}^\infty) : z=\sum_{i=1}^\infty \varepsilon_ix_i, \ x_i\in X, \ \sum_{i=1}^\infty \varepsilon_i < \infty\}$. Note that in this form it is allowed to use an element of $X$ more than once. So just because $r_X(z)$ is non-empty, it does not mean that $z \in FS(X)$. 

We are going to look at shortest representations of $n \in \N$ for different sets. For this we will use the notation $rank_Y(p):= \min_{\varepsilon \in r_Y(p)}(\sum_i \varepsilon_i)$, i.e. the length of the shortest representation of $p$ from the set Y. If the set Y is obvious from the context, then we leave the Y from the index and simply write $rank(p)$. Denote by $mult_Y(p,\varepsilon)$ the maximal multiplicity of an element in the representation $\varepsilon$ from $r_Y(p)$. For example if $p=a_1+a_2+a_2+a_3+a_3+a_3$ ($\varepsilon_1=1, \varepsilon_2=2, \varepsilon_3=3$) then $mult_A(p,\varepsilon)=3$ (sometimes we will leave the $A$ from the index, just like in the case of the rank).
A sequence $A$ is said to be {\it regular} if all numbers $n$ have a shortest representation which has multiplicity equal to 1. Formally for all $n\in N$ when we assume that $rank_A(n)=k$ then $\ min_{\varepsilon \in r^k_A(n)}(mult_A(n,\varepsilon))=1$, where $r^k_A(n)=\{ (\varepsilon \in r_A(n) | \sum_i \varepsilon_i = k \}$.

Denote $[n]:=\{1,2,\dots n\}$ and let $B_r(N)$ be the set $B_r(N)=\{(x_1,x_2,\dots, x_r)\in \N^r: \ 1\leq x_i \leq N; \ i=1,2,\dots ,r\}$. When $r=2$ we write shortly $B(N)$.


\smallskip

In the next Lemma we will show that  sequences which fulfill conditions (i) and (ii) are complete. 
\begin{lemma}\label{1.1}
Let $Y\subseteq \N$ and assume that $1\in Y$ and $Y\setminus \{1\}\subseteq Y+Y$. Then $Y$ is complete. Moreover if $Y=\{1=y_1<y_2<\dots\}$ then for every $i=1,2,\dots$ we have $y_i\leq 2^{i-1}$.
\end{lemma}
\begin{proof}[Proof of Lemma~\ref{1.1}]

Write $Y=\{1=y_1<y_2<\dots <y_n<\dots \}$. Since $Y\subseteq Y+Y$ we have that $1,2\in FS(Y)$ and for every $n>1$, $y_n=y_i+y_j$. Clearly $i,j\leq n-1$, so we have that $y_n=y_i+y_j\leq 2y_{n-1}$.


From this point the proof is well-known; for the sake of completeness we include the rest of the argument. We claim that for every $k\in \N$: $[2a_k-1]\subseteq FS(Y)$, and each element of $[2a_k-1]$ is in the form $ \{ \sum\limits_{i=1}^k \varepsilon_i a_i \mid \varepsilon \in \{0,1\} \}$ . It is true for $k=1$ and assume it is true for $k\geq 2$. Since $a_{k+1} \leq 2a_k$, hence all positive integers are in $FS(Y)$ up to $a_{k+1}+(2a_k -1)\geq  2a_{k+1}-1$. So obviously $[2a_{k+1}-1]\subseteq FS(Y)$.

\end{proof}


Note that many 'classical' sequences which fulfill conditions above (e.g. the sequence of two powers, the Fibonacci sequence) are also regular.

\begin{lemma}
The sequences of two powers and the Fibonacci numbers ($F_1=1, \ F_2=2, \dots$) are regular and fulfills conditions (i)-(iii).
\end{lemma}
\begin{proof}
That the sequences of two powers and the Fibonacci numbers fulfills conditions (i)-(iii) is clear.

The regularity assumption for the shortest representation obviously holds for $\{2^k\}_{k=0}^{\infty}$, since the unique one is such a representation.

We show that it also holds for the Fibonacci numbers. Assume, that a shortest representation of $n = \sum_{k=1}^n F_{i_k}$ contains $F_{i_s}$ twice. If $F_{i_s}=1$ then by replacing $2 F_{i_s}$ with $F_2=2$ we get a shorter representation which is not possible. If $F_{i_s}=2$, then we replace $2F_{i_s}$ with $F_1+F_3$. Otherwise we have that $F_{i_s}+F_{i_s}=F_{i_s-2}+F_{i_s+1}$. We continue this process (replacing duplications the way described above), until there is no duplication. Since we can only have one $F_1$ in the current representation at any time (otherwise we could get a shorter one) the case when $F_2$ is duplicated will occur at most once. So with at most one exception the sum of the indexes decreases by one at every step of the process. So eventually it will stop and in the end we always acquire a representation with $mult_A(\varepsilon,n)=1$ which is not longer than the one we started with.
\end{proof}

\begin{rem}
Let us note that there are sequences which ensure conditions $(i)-(ii)$ but they are not regular.

Let $A=\{1,2,3,5,6,12,\dots \}$. This sequence fulfils conditions $(i)$ and $(ii)$. Here $10=2+3+5\in FS(A)$ (an it is the shortest representation with multiplicity=1), but $10=5+5$ is another (shorter) representation. 
\end{rem}



Now we turn onto the communication problem. In the next theorem we will use number-in-hand multiparty communication model, i.e. there are $k$ players $P_1,P_2,\dots P_k$ and a $k$-argument functions $F: (\bo)^k\mapsto \{0,1\}$. For every $i\in [k]$ $P_i$ gets an $n$-bit input. In the communication process we will use {\it blackboard model} where every message sent by a player is written down on a blackboard which is visible for all players.

The communication complexity of
this model, denoted by $CC^{(k)}(F)$, is the least number of bits needed to be communicated to compute $F$ correctly.

Assume that we have $k$ players and we assign a regular sequence $A_i$ to each of them. For a given point $p=(p_1,p_2,\dots, p_k)\in \N^k; \ p_i\leq N;  \ (i=1,2,\dots,k)$, the $i^{th}$ players knows (just) $p_i$ and his previously given set $A_i$. Let $X:=A_1 \times \dots \times A_k$. With minimal communications they have to decide whether $p\in FS(X)$ or not. Denote by $F$ the function which describes this.

\smallskip

We will prove in Section 2 that:

\begin{thm}\label{The}
Let  $X=A_1\times A_2\times \dots \times A_k\subseteq \N^k$, where for every $i=1,2,\dots k$ $A_i$ is regular and $(i)$, $(ii)$ and $(iii)$ hold. Then
$$
CC^{(k)}(F)<k\log_2\Big(\frac{\log_2N}{\log_2\varrho}\Big)+k.
$$
\end{thm}

\bigskip

In the rest of the paper we will investigate additive structure of special $X$ sets. The next proposition we show that except a region with zero density all lattice points are in $FS(\{2^m\}\times\{2^k\})_{m,k\in \N}$. We say that $E:=\{(a,b)\in \N^2: \ b\leq \log_2 a\}\cup \{(a,b)\in \N^2: \ a\leq \log_2 b\}$ is the exceptional set.  We prove the following two propositions in Section 3:

\begin{prop}\label{prop}
Let $X=\{2^m\}\times\{2^k\}_{m,k\in \N}$ and let $E$ be the exceptional set. Then

1. $(\N^2\setminus E)\subseteq FS(X)$.

Furthermore

2. For every $D\in \N$ there exists a square $S_D:=\{(s_1,s_2): \ x_0\leq s_1\leq x_0+D; y_0\leq s_2\leq x_0+D \ \}\subseteq E$ (for some $x_0,y_0\in \N$) such that $FS(X)\cap S_D=\emptyset$.
\end{prop}



Nevertheless the set $E$ is not "empty". It contains "many" lattice points from $FS(X)$:

\begin{prop}\label{1.3}
For every $M\in \N$ there exists a square $S_M:=\{(t_1,t_2): \ z_0\leq t_1\leq z_0+M; w_0\leq t_2\leq w_0+M \ \}\subseteq E$ (for some $z_0,w_0\in \N$) such that $$|FS(X)\cap S_M|\geq \frac{1}{4}M\log_2M.$$
\end{prop}

\smallskip



\section{Proof of Theorem \ref{The}}

First we prove this simple but crucial lemma:

\begin{lemma}\label{2.1}
Assume that $A_1,A_2\subseteq \N$ are complete and regular sequences, and assume that $(p_1,p_2)\in FS(A_1)\times  FS(A_2)$ for some positive integers $p_1\geq p_2$. Then $(p_1,p_2)\in FS(A_1\times A_2)$ if and only if $p_2\geq rank(p_1)$ and $p_1\geq rank(p_2)$.
\end{lemma}

\begin{proof}[Proof of Lemma \ref{2.1}]

First note that one of $p_2\geq rank(p_1)$ and $p_1\geq rank(p_2)$ is always satfified, because, for example, if $rank(p_1) \geq rank(p_2)$ then $p_1 \geq rank(p_1) \geq rank(p_2)$. So from now on, without the loss of generality, we assume $rank(p_1) \geq rank(p_2)$.

($\Leftarrow$): Let the representation of $p_1$ in $FS(A_1)$  and $p_2$ in $FS(A_2)$ be $p_1=x_{t_1}+x_{t_2}+\dots +x_{t_r}$ and  $p_2=y_{i_1}+y_{i_2}+\dots +y_{i_k}$ respectively.
Since every $y_{i_j}\in A_2 \setminus \{1\} \subseteq A_2+A_2$, we can split $y_{i_j}$ to the sum of two earlier element. Going on this process step by step we increase the number of terms by 1. This process can be continued until all terms in the sum are 1. Since $p_2\geq rank(p_1)$, there will be a step, where the number of (multi)terms is equal to $rank(p_1)$, i.e. $p_2=y_{j_1}+y_{j_2}+\dots +y_{j_s}$, $y_{j_1}\leq y_{j_2}\leq \dots \leq y_{j_s}$ and  where $s=rank(p_1)$. Now the points $(x_{t_i},y_{j_i})$, $1\leq i\leq s$ are pairwise different, hence
$$
(p_1,p_2)=\sum_{i=1}^s(x_{t_i},y_{j_i})\in FS(A_1\times A_2)
$$
as we stated.

($\Rightarrow$) Assume that $p_2<rank(p_1)$. We know that every representation of $p_2$ is not longer than $p_2$ and every representation of $p_1$ is not shorter than $rank(p_1)$. This means, that there are no representations of $p_1$ and $p_2$, where the elemenst of the two sums could be paired. Thus $(p_1, p_2) \notin FS(A_1 \times A_2)$.

\end{proof}

\medskip

{\it The protocol and the Calculation of $CC^{(k)}(F)$}:

\smallskip

\begin{proof}[Proof of Theorem 1.3.]
The point $p=(p_1,p_2,\dots, p_k)\in \N^k; \ p_i\leq N;  \ (i=1,2,\dots,k)$ is given. Recall that $P_i$ knows only $p_i$. 
Now $P_i$ can write $p_i$ in a shortest representation from $A_i$ as $p_i=a_{i_1}+a_{i_2}+\dots +a_{i_{t_i}}\in FS(A_i)$. Recall that since all $A_1,A_2,\dots ,A_k$ are regular hence they can do it in a way that for every $i$ these representations have $mult_{A_i}(p_i, \varepsilon)=1$.  Then $P_i$  sends $rank(p_i)$ to the list (to the blackboard). 

All player $P_i$ rearranges the list $\{rank(p_i)\}_{i=1}^k$, and finds the maximal $\max_i rank(p_i)$. Then player $P_i$ sends the bit $1$ if $p_i\geq \max_i rank(p_i)$ and $0$ otherwise. Clearly if the blackboard does not contain $0$ then the point $p=(p_1,p_2,\dots, p_k)$ lies in $FS(X)$ (the same argument can be applied as in the proof of Lemma 2.1.). 

Now we have to show that otherwise $p$ is not representable. 

Indeed, if there is a $0$ in the blackboard then there are two players $P_i$ and $P_j$ where the length of all representation of $p_j$ is longer than $p_i$. By Lemma \ref{2.1} taking the projection to the plane $i,j$ of our sets $X$ we get that the point $(p_i,p_j)$ is not representable.

\smallskip

{\it Calculation of $CC^{(k)}(F)$}: We give an upper bound to $CC^{(k)}(F)$. There are $k$ players, so it is enough to bound length of individual message of any player. 

Represent $p_j$ as $p_j=a_{i_1}+a_{i_2}+\dots +a_{i_{t_j}}$. Now by $(iii)$  $a_{i_{t_j}}\geq \varrho^{i_{t_j}}$ and since $a_{i_{t_j}}\leq p_j\leq N$ we get $t_j \leq i_{t_j} \leq \frac{\log_2N}{\log_2\varrho}$.
So the binary length of the message is at most
$
\log_2\Big(\frac{\log_2N}{\log_2\varrho}\Big)
$.
Finally the total number of bits is at most
$$
k\log_2\Big(\frac{\log_2N}{\log_2\varrho}\Big)+k.
$$
\end{proof}

\section{On the additive structure of $FS(\{2^m\}\times\{2^k\})_{m,k\in \N}$}

We devote this section to prove Proposition 1.4. and Proposition \ref{1.3}.

\begin{proof}[Proof of Proposition~\ref{prop}]

1. Since the set $X$ is reflected to the line $y=x$ so we will show that all $(a,b)\in \N^2\setminus E$, $b\leq a$ lies in $FS(X)$. Denote by $s(a)$ and $s(b)$ the number of terms in the dyadic expansion of $a$ and $b$ respectively. If $s(b)>s(a)$ then $s(b) < b \leq a$ and by Lemma \ref{2.1} we are done. 

Assume that  $s(b)<s(a)$. Clearly $s(a)\leq \log_2 a\leq b$, so we can use Lemma \ref{2.1} again to finish the proof.

2. First let us note that the shortest binary representation (including repetition as well) is the unique one (indeed if there are two equal terms, say $2^i-2^i$ then one can replace it by $2^{i+1}$ shortening the number of terms).

Let $D\in \N$ and for the square $S_D$ we define first the bottom-left corner $(x_0,y_0)$ of $S_D$.
Let $y_0:=1$ and $x_0:=\sum_{i=D+1}^{2D}2^i$. Now by Lemma \ref{2.1} we have that $(x_0+k,j) \notin FS(X)$ for every $1\leq k, j\leq D$ (because $rank(x_0+k) > D \geq j$), as we stated.
\end{proof}

\smallskip

\begin{proof}[Proof of Proposition~\ref{1.3}]  

Let us choose $R$ such that $2^R\leq M< 2^{R+1}$ holds. We simply define the square to be:
$$
S_M:=\{(t_1,t_2)\in \N^2: \ 2^{2^{R+1}}\leq t_1\leq 2^{2^{R+1}}+M; 0\leq t_2 \leq M  \}
$$
It is easy to see that $S_M \subset E$. We will only deal with the following subrectangle of $S_M$, because it will have still "enough" elements from $FS(X)$:
$$
S_{2^R-1,2^R}:=\{(t_1,t_2)\in \N^2: \ 2^{2^{R+1}}\leq t_1\leq 2^{2^{R+1}}+2^R-1; 0\leq t_2\leq 2^R \}
$$
Take any element from $FS(X)$ which is a sum of horizontal elements, i.e. $(n,m)=(2^{i_1},2^f)+(2^{i_2},2^f)+\dots +(2^{i_k},2^f)$, here $k=s(n)$ and denote this set by $Z$, more formally $Z:=\{(n,s(n)2^f): f \in \mathbb{N} \}$. We will show that $Z \cap S_{2^R-1,2^R}$ still has a lot of elements. By the choice of $S_{2^R-1}$ we have that for  every element $(n,m) \in S_{2^R-1,2^R}$: $s(n) \in \{1, \dots R+1 \}$, moreover, we have that for a fixed $k \in \{1, \dots R+1 \}$:
$$
\mid \{ n:  2^{2^{R+1}}\leq n \leq 2^{2^{R+1}}+2^R-1 \text{ and } s(n) = k \} \mid = \binom{R}{k-1}. 
$$

So if we return to the set $Z$, to ensure that an $(n,m)=(n,s(n)2^f) \in Z$ is in $S_{2^R-1,2^R}$ we need that $s(n)2^f \leq 2^R$. So $f$ should be the element of the following set: $\{0,1, \dots , \lfloor R-\log_2 s(n) \rfloor \} $. Hence we get that the number of intersections between the sets $Z$ and $S_{2^R-1,2^R}$ is at least: 
$$\sum\limits_{k=1}^{R+1} \binom{R}{k-1} (\lfloor R-\log_2 k \rfloor +1 )  \geq \sum\limits_{k=0}^{R} \binom{R}{k} ( R-\log_2 (k+1)  )$$

Now if we use that $\frac{R}{2} \geq \log_2(R+1)$, if $R \geq 6$, we can continue the estimation:

$$\sum\limits_{k=0}^{R} \binom{R}{k} ( R-\log_2 (R+1) ) \geq \sum\limits_{k=0}^{R} \binom{R}{k} \frac{R}{2}= 2^R \frac{2R}{4} \geq \frac{M}{2} \frac{\log_2 M}{2}.$$

\end{proof}

\section{On non-regular case}

Let $X=A_1\times A_2\subseteq \N^2$. As we have seen there are sequences, where the rank of an element of $FS(A_i)$ ($i=1,2$) is not the length of the shortest representation. It is not obvious that which conditions ensure that a given point $(p_1,p_2)$ is an element of $FS(X)$ or not. Clearly it is necessary to have (multi)partitions of $p_1=y_{t_1}+y_{t_2}+\dots +y_{t_r}$ and  $p_2=y_{i_1}+y_{i_2}+\dots +y_{i_k}$ for which $r=k$. But it is not sufficient.

Let $A_1=A_2=\{1,2,3,6,12,13,26,52,\dots,2^t13,\dots\}$ and  consider $(3,50)\in \N^2$. We have $3=1+1+1=1+2$ furthermore $50=26+12+12$, and it is easy to check that the number of terms in any other representation of $50$ is at least four, i.e. there is no pairwise different matching of the elements although there is a representation where the number of terms are the same. So $(3,50)\not \in FS(X)$.

Nevertheless if $X=A_1\times A_2$, where $A_1$ and $A_2$ fulfill conditions $(i),(ii)$ and $(iii)$, we will show an additional condition which is enough.


\begin{prop}
Let  $A_1$ and $A_2$ fulfill conditions $(i),(ii)$ and $(iii)$ and let $X=A_1\times A_2$. Let $p_1,p_2\in \N$ with $rank_{A_1}(p_1)=rank(p_1)$ and $rank_{A_2}(p_2)=rank(p_2)$ (for simplicity), and let $\varepsilon_1$ and $\varepsilon_2$ be shortest representations of $p_1$ and $p_2$. Let $K:=\max\{mult(p_1, \varepsilon_1), mult(p_2, \varepsilon_2), |rank(p_1)-rank(p_2)| \}$ and $L:=\min\{rank(p_1),rank(p_2)\}$.
If there are $\varepsilon_1$ and $\varepsilon_2$ for $p_1$ and $p_2$ such that, $K\leq \sqrt{L}/2$, then $(p_1,p_2)\in FS(X)$.
\end{prop}

\begin{proof}
We start the proof with an important lemma which gives a sufficiently condition for the pairwise different matching of the co-ordinates:

\begin{lemma}[{\bf [1, Proposition 1]}]\label{3.5} Let $X_1, \dots , X_s$ be disjoint
 finite sets and $Y_1, \dots , Y_t$ be disjoint
finite sets too. Let

$$U=\bigcup_{i=1}^s X_i,\quad V= \bigcup_{j=1}^tY_j,$$
with $|U|=|V|$ and suppose that
$1\le |X_i|\le \sqrt{|U|}$ for $i=1,2,\dots s$ and $1\le |Y_j|\le \sqrt{|V|}$ for $j=1,2,\dots t$.   
Then there exists a bipartite graph $G(X,Y)$ fulfilling the
following conditions:

\begin{itemize}
\item[(a)] there are no two edges $(x_1,y_1); (x_2,y_2)$ for which
$x_1,x_2\in X_i; \ y_1,y_2\in Y_j$ for some $i$ and $j$;
\item[(b)]  $G(X,Y)$ is a matching.
\end{itemize}
\end{lemma}

We have $mult(p_i, \varepsilon_i)\leq K$ and $K \leq \frac{\sqrt{rank(p_i)}}{2}$ (for $i=1,2$), because of the definitions of $K$ and $L$. Hence the number of distinct terms in the shortest representations are at least $2\sqrt{rank(p_i)}$, $i=1,2$. Assume now w.l.o.g. that
\begin{equation}\label{2}
0 \leq rank(p_1)-rank(p_2)\leq K \leq \frac{\sqrt{rank(p_2)}}{2}.
\end{equation}
Hence by the splitting process (used previously in Lemma 2.1. and Theorem 1.3.) there is a (multi)partitions of $p_2$ in the form $p_2=n_1y'_1+n_2y'_2+\dots n_ry'_r$ 
with $\sum_j n_j = rank(p_1)$. By (\ref{2}) this process exists, furthermore the maximal number of repetition of element in the representation of $p_2$ is at most $\sqrt{rank(p_2)}$. 

Now in the representation of $p_1=k_1 x_1+ \dots +k_i x_i+ \dots +k_t x_t$ identify $k_i x_i$ to a set $X_i$ with cardinality $k_i$ ($i=1,2, \dots ,t$) and in the representation $p_2=n_1 y'_1+n_2 y'_2+ \dots + n_r y'_r$ identify $n_j y'_j$ to a set $Y_j$ with cardinality $n_j$ ($j=1,2, \dots,r$). 

Finally let us use Lemma 4.2 which ensures that $(p_1,p_2)=\sum_{j=1}^{rank(p_1)}(x_{s_j},y'_{t_j})$, where the $(x_{s_j},y'_{t_j})$ pairs form a pairing according to the lemma. So  $(p_1,p_2)\in FS(A_1\times A_2)$.

\end{proof}

\section{Concluding remarks}

Recall that the general knapsack problem is known to be NP-complete and sounds as follows: for a given sequence $A=\{a_1,a_2,\dots, a_n\}\subset \N$ decide that the equation $s=\sum_{i=1}^n\varepsilon_ia_i; \ \varepsilon_i\in\{0,1\}$, $i=1,2,\dots, n$ is solvable or not in $\varepsilon_1,\varepsilon_2,\dots,\varepsilon_n$.

The density of a knapsack problem is defined as: $d:=\frac{n}{\log_2(\max a_i)}$. When $d<1$ then there is a possible encryption process. When $d>1$  there is no an effective 
approach to attack the knapsack problem. The main tool is the so called {\it }basis reduction method.


Now we will show a way to reduce this problem, decide whether a given point $(p_1,p_2)$ is an element of $FS(A_1\times A_2)$ or not, to a classical knapsack problem.



Let $(p_1,p_2)\in \N^2$ and assume that $1\leq p_1,p_2\leq M$. Let $B_1:=\{x_1<x_2<\dots x_k\}\subset A_1$ and $B_2:=\{y_1<y_2<\dots y_m\}\subset A_m$, where $k=\max\{T:x_1+x_2+\dots +x_T\leq M\}$ and $m=\max\{R:y_1+y_2+\dots +y_R\leq M\}$. 

Let now $Z:=\{z=Mx_i+y_j:1\leq i\leq k; \ 1\leq j\leq m\}\subset \N$. Observe that $(p_1,p_2)\in FS(A_1\times A_2)$ if and only if $Mp_1+p_2\in FS(Z))$.


\smallskip

\bigskip

\noindent{\bf Acknowledgement.} The second named author is supported by grant K-129335. The first and third named authors are supported by the European Union, co-ﬁnanced by the European Social
Fund (EFOP-3.6.3-VEKOP-16-2017-00002).

\end{document}